\documentclass[runningheads]{llncs}
\usepackage{amsmath,amsfonts}
\usepackage{amssymb}
\usepackage{mathtools}
\usepackage{booktabs}
\usepackage{hyperref}
\usepackage[ruled,vlined]{algorithm2e}
\usepackage{graphicx}
\usepackage[colorinlistoftodos]{todonotes}
\usepackage{nicefrac}
\usepackage{diagbox}
\usepackage[misc]{ifsym}

\usepackage{caption}
\usepackage[list=true]{subcaption}
\usepackage{tabularx}
\newcolumntype{Y}{>{\centering\arraybackslash}X}
\usepackage{pgfplots}
\pgfplotsset{width=10cm,compat=1.9}
\usepackage{siunitx}

\begin{document}
\title{Estimating Descriptors for Large Graphs\thanks{The first two authors have been supported by the grant received to establish CIPL and the third author has been supported the grant received to establish SEIL, both associated with the National Center in Big Data and Cloud Computing, funded by the Planning Commission of Pakistan.}}
%
%
\author{Zohair Raza Hassan\inst{1} \and
Mudassir Shabbir\inst{1} \and
Imdadullah Khan\inst{2}(\Letter)
\and Waseem~Abbas\inst{3}}

\authorrunning{Z.R. Hassan, M. Shabbir, I. Khan, and W. Abbas}
\institute{Information Technology University of the Punjab, Pakistan \\
\email{zohair.raza@itu.edu.pk, mudassir.shabbir@itu.edu.pk} \\
\and
Lahore University of Management Sciences, Pakistan
\email{imdad.khan@lums.edu.pk} \\
\and
Vanderbilt University, USA
\email{waseem.abbas@vanderbilt.edu}}

\maketitle              
\begin{abstract}
Embedding networks into a fixed dimensional feature space, while preserving its essential structural properties is a fundamental task in graph analytics. These feature vectors (graph descriptors) are used to measure the pairwise similarity between graphs. This enables applying data mining algorithms (e.g classification, clustering, or anomaly detection) on graph-structured data which have numerous applications in multiple domains. State-of-the-art algorithms for computing descriptors require the entire graph to be in memory, entailing a huge memory footprint, and thus do not scale well to increasing sizes of real-world networks. In this work, we propose streaming algorithms to efficiently approximate descriptors by estimating counts of sub-graphs of order $k\leq 4$, and thereby devise extensions of two existing graph comparison paradigms: the Graphlet Kernel and NetSimile. Our algorithms require a single scan over the edge stream, have space complexity that is a fraction of the input size, and approximate embeddings via a simple sampling scheme. Our design exploits the trade-off between available memory and estimation accuracy to provide a method that works well for limited memory requirements. We perform extensive experiments on real-world networks and demonstrate that our algorithms scale well to massive graphs.

\keywords{Graph Descriptor \and Edge Stream \and Graph Classification}
\end{abstract}

\section{Introduction}
\label{sec:intro}

Evaluating similarity or distance between a pair of graphs is a building block of many fundamental data analysis tasks on graphs such as classification and clustering. These tasks have numerous applications in social network analysis, bioinformatics, computational chemistry, and graph theory in general. Unfortunately, large orders (number of vertices) and massive sizes (number of edges) prove to be challenging when applying general-purpose data mining techniques on graphs. Moreover, in many real-world scenarios, graphs in a dataset have varying orders and sizes, hindering the application of data mining algorithms devised for vector spaces. Thus, devising a framework to compare graphs with different orders and sizes would allow for rich analysis and knowledge discovery in many practical domains.

However, graph comparison is a difficult task; the best-known solution for determining whether two graphs are structurally the same takes quasi-polynomial time~\cite{babai2016graph}, and determining the minimum number of steps to convert one graph to another is \textsc{NP-Hard}~\cite{sanfeliu1983distance}.  In a more practical approach, graphs are first mapped into fixed dimensional feature vectors, where vector space-based algorithms are then employed. In a supervised setting, these feature vectors are learned through neural networks~\cite{morris2019weisfeiler,wu2019comprehensive,xu2018powerful}. In unsupervised settings, the feature vectors are descriptive statistics of the graph such as average degree, the eigenspectrum, or spectra of sub-graphs of order at most $k$ contained in the graph~\cite{sge,kondor2016multiscale,shervashidze2011weisfeiler,shervashidze2009efficient,tsitsulin2018netlsd,verma2017hunt}. 

The runtimes and memory costs of these methods depend directly on the magnitude (order and size) of the graphs and the dimensionality (dependent on the number of statistics) of the feature-space. While computing a larger number of statistics would result in richer representations, these algorithms do not scale well to the increasing magnitudes of a real-world graphs~\cite{faloutsos2013deltacon}.

A promising approach is to process graphs as streams - one edge at a time, without storing the whole graph in memory. In this setting, the graph descriptors are approximated from a representative sample achieving practical time and space complexity~\cite{Chen:2017:UFE:3110025.3110042,Sanei-Mehri:2019:FBE:3357384.3357983,shin2017wrs,shin2018tri,stefani2017triest}.  


In this work we propose \textsc{gabe} (Graphlet Amounts via Budgeted Estimates), and \textsc{maeve} (Moments of Attributes Estimated on Vertices Efficiently), stream-based extensions of the Graphlet Kernel~\cite{shervashidze2009efficient}, and NetSimile~\cite{berlingerio2013network}, respectively. Our contributions can be summarised as follows: \begin{itemize}
    \item We propose two simple and intuitive descriptors for graph comparisons that run in the streaming setting.
    \item We provide analytical bounds on the time and space complexity of our feature vectors generation; for a fixed budget, the runtime and space cost of our algorithms are linear.
    \item We perform extensive empirical analysis on benchmark graph classification datasets of varying magnitudes. We demonstrate that \textsc{gabe} and \textsc{maeve} are comparable to the state-of-the-art in terms of classification accuracy, and scale to networks with millions of nodes and edges.
\end{itemize}

The rest of the paper is organized as follows. We discuss the related work in Section~\ref{sec:rw}. Section~\ref{sec:prelim} discusses all preliminaries required to read the text. We present \textsc{gabe} and \textsc{maeve} in Section~\ref{sec:sol}.
We report our experimental findings in Section~\ref{sec:experiments} and finally conclude the paper in Section~\ref{sec:conclusion}.

\section{Related Work}
\label{sec:rw}

Methods for comparing a pair of graphs can broadly be categorized into \textit{direct approaches}, \textit{kernel methods}, \textit{descriptors}, and \textit{neural models}. 
Direct approaches for evaluating the similarity/distance between a pair of graphs preserve the entire structure of both graphs. The most prominent method under this approach is the {\em Graph Edit Distance} (\textsc{ged}), which counts the number of edit operations (insertion/deletion of vertices/edges) required to convert a given graph to another~\cite{sanfeliu1983distance}. Although intuitive, \textsc{ged} is stymied by its computational intractability. 
Computing distance based on the vertex permutation that minimizes the ``error'' between the adjacency representations of two graphs is a difficult task~\cite{babai2016graph}, and proposed relaxations of these distances are not robust to permutation~\cite{bento2018family}.
An efficient algorithm for large network comparison \textsc{DeltaCon}, is proposed in~\cite{faloutsos2013deltacon} but it is only feasible when there is a valid one-to-one correspondence between vertices of the two graphs.

In the kernel-based approach, graphs are mapped to a fixed dimensional vector space based on various substructures in the graphs. A kernel function is then defined, which serves as a pairwise similarity measure that takes as input a pair of graphs and outputs a non-negative real number. Typically, the kernel value is the inner-product between two feature vectors corresponding to the two graphs. This so-called kernel trick has been used successfully to evaluate pairwise of  other structures such as images and sequences~\cite{Bo_ImageKernel,Kuksa_SequenceKernel,Farhan_SequenceKernel}. Several graph kernels based on sub-structural patterns have been proposed, such as the Shortest-Path~\cite{borgwardt2005shortest} and Graphlet~\cite{shervashidze2009efficient} kernels.  More recently, a hierarchical kernel based on propagating spectral information within the graph~\cite{kondor2016multiscale} was introduced. The WL-Kernel~\cite{shervashidze2011weisfeiler} that is based on the Weisfeller-Lehman isomorphism test has been shown to provide excellent results for classification and is used as a benchmark in the graph representation learning literature. Kernels require expensive computation and typically necessitate storing the adjacency matrices, making them infeasible for massive graphs.

Graph Neural Networks (\textsc{gnn}s) learn graph level embeddings by aggregating node representations learned via convolving neighborhood information throughout the neural network's layers. This idea has been the basis of many popular neural networks and is as powerful as WL-Kernels for classification~\cite{morris2019weisfeiler,xu2018powerful}. We refer interested readers to a comprehensive survey of these models~\cite{wu2019comprehensive}. 
Unfortunately, these models also require expensive computation and storing large matrices, hindering scalability to real-world graphs.

Graph descriptors, like the above two paradigms, attempt to map graphs to a vector space such that similar graphs are mapped to closely in the Euclidean space. Generally, the dimensionality of these vectors is small, allowing efficient algorithms for graph embeddings. NetSimile~\cite{berlingerio2013network} describes graphs by computing moments of vertex features, while SGE~\cite{sge} uses random walks and hashing to capture the presence of different sub-structures in a graph.
State of the art descriptors are based on spectral information;~\cite{verma2017hunt} proposed a family of graph spectral distances and embedding the information as histograms on the multiset of distances in a graph, and NetLSD~\cite{tsitsulin2018netlsd} computes the heat (or wave) trace over the eigenvalues of a graph's normalized Laplacian to construct embeddings.

The fundamental limitation of all the above approaches is the requirement that the entire graph is available in memory. This limits the applicability of the methods to a graph of small magnitude. To the best of our knowledge, this work is the first graph comparison method that does not assume this.

Streaming algorithms assume an online setting; the input is streamed one element at a time, and the amount of space we are allowed is limited. This allows one to design scalable approximation algorithms to solve the underlying problems. There has been extensive work on estimating triangles (cycles of length three) in graphs~\cite{shin2017wrs,stefani2017triest}, butterflies (cycles of length four) in bipartite graphs~\cite{Sanei-Mehri:2019:FBE:3357384.3357983}, and anomaly detection~\cite{eswaran2018sedanspot} when the graph is input as a stream of edges. A framework for estimating the number of connected induced sub-graphs on three and four vertices is presented in~\cite{Chen:2017:UFE:3110025.3110042}. 

\section{Preliminaries and Problem Definition}
\label{sec:prelim}

\subsection{Notation and Terminology}

Let $G = \left(V_G, E_G\right)$ be an undirected, unweighted, simple graph, where $V_G$ is the set of vertices and $E_G$ is the set of edges. 

For  $v \in V_G$, let $N_G\left(v\right) = \{ u : \left(u, v\right) \in E_G \}$ be the set of neighbors of $v$, and $d^v_G := |N_G\left(v\right)|$ the degree of $v$. A graph is connected if and only if there exists a path between all pairs in $V_G$.

A sub-graph of $G$ is a graph, $G' = \left(V_{G'}, E_{G'}\right)$, such that $V_{G'} \subseteq V_G$ and $E_{G'}$ is a subset of edges in $E_G$ that are incident only on the vertices present in $V_{G'}$, i.e. $E_{G'} \subseteq \{ \left(u, v\right) : \left(u, v\right) \in E_G \land u, v \in V_{G'}\}$. If equality holds ($E_{G'}$ contains all edges from the original graph), then $G'$ is called an induced sub-graph of $G$.

Two graphs, $G_1$ and $G_2$, are isomorphic if and only if there exists a permutation $\pi: V_{G_2} \rightarrow V_{G_1}$ such that $E_{G_1} = \{ (\pi(u), \pi(v)) : (u,v) \in E_{G_2} \}$. For a graph $F = \left(V_{F}, E_{F}\right)$, let $H^{F}_G$ (resp. $\widehat{H}^{F}_G$) be the set of sub-graphs (resp. induced sub-graphs) of $G$ that are isomorphic to $F$.

We assume vertices in $V_G$ are denoted by integers in the range $[0,|V_G|-1]$. Let $S = e_1, e_2, \ldots, e_{|E_G|}$ be a sequence of edges in an arbitrary but fixed order, i.e. $e_t = \left(u_t, v_t\right)$ is the $t^{th}$ edge. Let $b$ be the maximum number of edges (budget) one can store in our sample, referred to as $\widetilde{E_G}$.

\subsection{Problem Definition}

We now formally define the graph descriptor problem:

\begin{problem}[Constructing Graph Descriptors]
Let $\mathcal{G}$ be the set of all possible undirected, unweighted, simple graphs. We wish to find a function, $\varphi: \mathcal{G} \rightarrow \mathbb{R}^d$, that can map any given graph to a $d$-dimensional vector.
\end{problem}

Existing work~\cite{berlingerio2013network,tsitsulin2018netlsd} on graph descriptors asserts that the underlying algorithms should be able to run on any graph, regardless of order or size, and should output the same representation for different vertex permutations. Moreover, the descriptors should capture features that can be compared across graphs of varying orders; directly comparing sub-graph counts is illogical as bigger graphs will naturally have more sub-graphs. The descriptors we propose are based on graph comparison methods that meet these requirements due to their graph-theoretic nature and feature scaling based on the graph's magnitude. We consider an online setting and model the input graph as a stream of edges. We impose the following constraints on our algorithms:

\begin{description}
    \item[C1:] \textbf{Single Pass:} The algorithm is only allowed to receive the stream once.
    \item[C2:] \textbf{Limited Space:} The algorithm can store a maximum of $b$ edges at once.
    \item[C3:] \textbf{Linear Complexity:} Space and time complexity of the algorithms should be linear (for fixed $b$) with respect to the order and size of the graph.
\end{description}

\subsection{Estimating Connected Sub-graph Counts on Streams}
\label{sec:est}
\begin{problem}[Connected Sub-graph Estimation on Streams]
\label{prob:subest}
Let $S$ be a stream of edges, $ e_1, e_2, \ldots, e_{|E_G|}$ for some graph $G = (V_G, E_G)$. Let $F = (V_F, E_F)$ be a small connected graph such that $|V_F| \ll |V_G|$. Produce an estimate, $N^F_G$, of $|H^F_G|$ while storing a maximum of $b$ edges at any given instant.
\end{problem}
Based on previous works on sub-graph estimation~\cite{shin2018tri,shin2017wrs,Chen:2017:UFE:3110025.3110042,stefani2017triest} the underlying recipe for algorithms that solve Problem~\ref{prob:subest} consists of the following steps:
\begin{itemize}
    \item For each edge $e_t \in S$, counting the instances of $F$ incident on $e_t$. For example, if $F$ is a triangle, then it amounts to counting the number of triangles an edge $e_t$ is part of.
    \item A sampling scheme through which we can compute the probability of detecting $F$ in our sample, denoted by $p^F_t$, at the arrival of the $t^{th}$ edge.
\end{itemize}

At the arrival of $e_t$, we increment our estimate of $|H^F_G|$ by $1/p^F_t$ for all instances of $F$ in our sample $\widetilde{E_G}$ that $e_t$ belongs to. The pseudocode is provided in Algorithm~\ref{alg:pseudo}. This simple methodology allows one to compute estimates whose expected values are equal to $|H^F_G|$:

\begin{theorem}
\label{thm:unbiased}
Algorithm~\ref{alg:pseudo} provides unbiased estimates: $\mathbb{E}[N^F_G] = |H^F_G|$.
\end{theorem}
\begin{proof}
For a sub-graph $h \in H^G_F$, let $X_h$ be a random variable such that $X_h = 1/p^F_t$ if $h$ is detected at the arrival of its last edge in the stream $e_t$, and 0 otherwise. Clearly, $N^F_G = \sum_{h \in H^F_G} X_h$, and $\mathbb{E}[X_h] = (1/p^F_t)\times p^F_t = 1$. Therefore,
\begin{equation*}
    \mathbb{E}\left[N^F_G\right] = \mathbb{E}\left[\sum_{h \in H^F_G} X_h\right] = \sum_{h \in H^F_G} \mathbb{E}\left[X_h\right] = \sum_{h \in H^F_G} 1 = \left|H^F_G\right|.
\end{equation*}
\end{proof}
At the arrival of $e_t$, counting only the sub-graphs that $e_t$ belongs to ensures that we do count the same sub-graph twice. In this work, we employ reservoir sampling~\cite{Vitter:1985:RSR:3147.3165}, which has been shown to be effective for sub-graph estimation~\cite{Chen:2017:UFE:3110025.3110042,shin2018tri,stefani2017triest}. Using reservoir sampling, the probability of detecting an $F$ that $e_t$ belongs to at the arrival of $e_t$ is equivalent to the probability that $|E_F| -1$ particular edges are present in the sample after $t-1$ time-steps: $p^F_t = \min\left(1, \prod^{|E_F| -2}_{i = 0} \frac{b - i}{t - 1 -i}\right)$.

\begin{algorithm}[t]
\SetAlgoLined
\DontPrintSemicolon
\SetKwInOut{Input}{Input}\SetKwInOut{Output}{Output}
\Input{Stream of edges $S = e_1, e_2, \ldots, e_{|E_G|}$, budget $b$, and a graph $F$}
\Output{$N^F_G$ (estimate of $|H^F_G|$)}
$\widetilde{E_G} \leftarrow \emptyset$, $N^F_G \leftarrow 0$ \tcc*[r]{Initialize sample of edges, and estimate}
\For{$e_t \in S$}{
    Find all instances of $F$ that $e_t$ belongs to in  $\widetilde{E_G} \cup \{e_t\}$\;
    Increment $N^F_G$ by $1/p^F_t$ for each $F$ detected\;
    Sample $e_t$ in $\widetilde{E_G}$, based on $b$
}
\caption{Sub-graph Estimation on Streams}
\label{alg:pseudo}
\end{algorithm}

We now derive an upper bound for the variance. Note that while the bound is loose, it is sufficient to show that we obtain better results with greater $b$, and applies to any connected graph, $F$.

\begin{theorem}
\label{thm:var}
When using reservoir sampling, the variance of $N^F_G$ in Algorithm~\ref{alg:pseudo} is bounded as follows: $\mathrm{Var}[N^F_G] \leq |H^F_G|^2 \prod^{|E_F| -2}_{i = 0} \frac{|E_G| - i}{b -i}$.
\end{theorem}
\begin{proof}
The theorem is trivially true when $b \geq |E_G| - 1$. We now explore the case when $b < |E_G| - 1$. Let $X_h$ be a random variable as defined in the proof for Theorem~\ref{thm:unbiased}. Note that $p^F_t \geq p^F_{|E_G|}$, and $\mathrm{Var}[X_h] = \mathbb{E}[X_h^2] - \mathbb{E}[X_h]^2 = 1/p^F_t - 1 \leq 1/p^F_{|E_G|}$. We bound the total variance using the Cauchy-Schwarz inequality:
\begin{align*}
    \mathrm{Var}[N^F_G] &= \sum_{h \in H^F_G} \sum_{h' \in H^F_G} \mathrm{Cov}[X_h,X_{h'}] \leq \sum_{h \in H^F_G} \sum_{h' \in H^F_G} \sqrt{\mathrm{Var}[X_h] \mathrm{Var}[X_{h'}]} \\
    &\leq \sum_{h \in H^F_G} \sum_{h' \in H^F_G} \frac{1}{p^F_{|E_G|}} = |H^F_G|^2 \prod^{|E_F| -2}_{i = 0} \frac{|E_G| - 1 - i}{b -i}.
\end{align*}
\end{proof}
Note that this methodology is also applicable for estimating the number of sub-graphs that each vertex is incident in, and simple modifications to the proofs for Theorems~\ref{thm:unbiased} and \ref{thm:var} will prove the same results for estimations on vertex counts.
\section{\textsc{GABE} and \textsc{MAEVE}}\label{sec:sol}
In this section discuss our two proposed descriptors: Graphlet Amounts via Budgeted Estimates (\textsc{gabe}), which is based on the Graphlet Kernel, and Moments of Attributes Estimated on Vertices Efficiently (\textsc{maeve}), based on NetSimile. 

\begin{figure}[t]
    \centering
    \includegraphics[width=.9\textwidth]{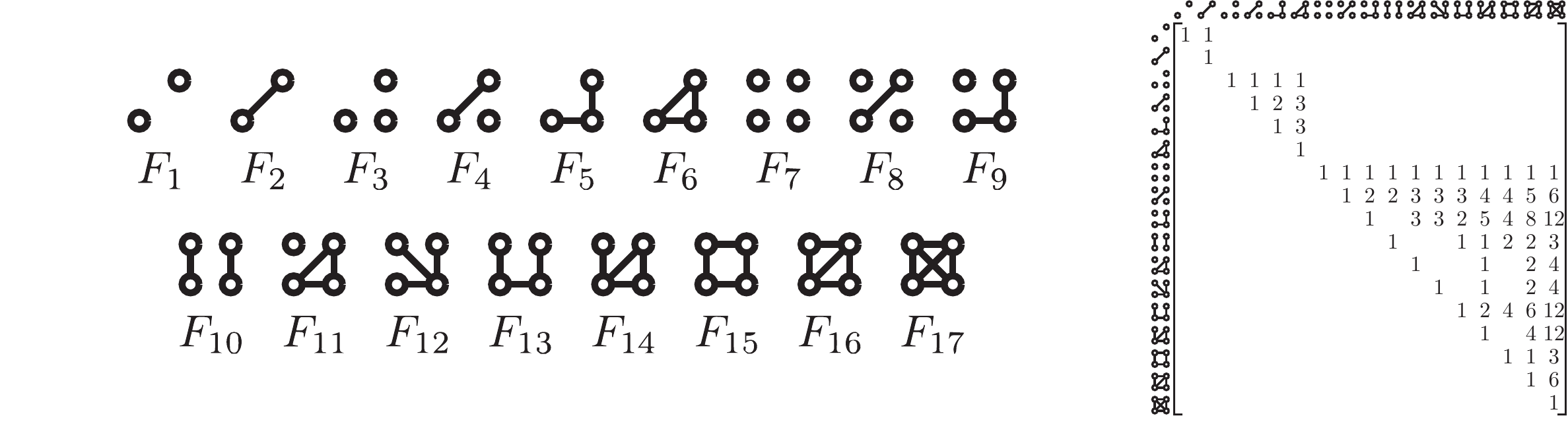}
    \caption{The graphs counted by \textsc{gabe}, and their corresponding overlap matrix $\mathcal{O}$ (best viewed when zoomed in). Zeros have been omitted for readability.}
    \label{fig:gabe1}
\end{figure}
\subsection{Graphlet Amounts via Budgeted Estimates}
\label{sec:gabe}
Let $\mathcal{F}_k$ be the set of graphs with order $k$. For two given graphs, $G_1$ and $G_2$, Shervashidze et al.~\cite{shervashidze2009efficient} propose counting all graphlets (induced sub-graphs) of order $k$ in both graphs, and computing similarity based on the inner product $\langle \phi_k(G_1), \phi_k(G_2) \rangle$, where, for a given $k$, and graphs $F_i \in \mathcal{F}_k$: 

\begin{equation*}
\phi_k(G) := \frac{1}{{\binom{|V_G|}{k}}} \begin{bmatrix}
\left|\widehat{H}^{F_1}_G\right| & \left|\widehat{H}^{F_2}_G\right| & \left|\widehat{H}^{F_3}_G\right| & \cdots & 
\left|\widehat{H}^{F_{|\mathcal{F}_k|-1}}_G\right| & 
\left|\widehat{H}^{F_{|\mathcal{F}_k|}}_G\right|
\end{bmatrix}^\intercal
\end{equation*}

Their algorithm runs in $O(|V_G|d^{k-1})$ ($d = \max_{v \in V_G} d^v_G$) for $k \in \{3,4,5\}$, and uses adjacency matrices. We use the methodology of ~\cite{Chen:2017:UFE:3110025.3110042}, to estimate the sub-graph counts as in Section~\ref{sec:est}, then compute induced sub-graph counts based on the overlap of graphs of the same order. We follow this procedure for estimating sub-graph counts of order $k \in \{2,3,4\}$, then concatenate the resultant $\phi_k(G)$'s into a vector. The 17 graphs we enumerate are shown in Figure~\ref{fig:gabe1}. Note that unlike~\cite{Chen:2017:UFE:3110025.3110042}, we also estimate the counts of disconnected induced sub-graphs.

\vskip.05in 
\noindent{\bf Induced Sub-graph Counts.} Let $\mathcal{F} = \{F_1, F_2, \ldots, F_{17}\}$ be the set of graphs we enumerate. Let $\mathcal{H}^{\mathcal{F}}_G$ (resp. $\widehat{\mathcal{H}}^{\mathcal{F}}_G$) be a vector such that $i^{th}$ entry corresponds to $|H^{F_i}_G|$ (resp. $|\widehat{H}^{F_i}_G|$).
Let $\mathcal{O}$ be a $|\mathcal{F}| \times |\mathcal{F}|$ matrix such that $O(i,j)$ is the number of sub-graphs of $F_j$, isomorphic to $F_i$, when $|V_{F_i}| = |V_{F_j}|$, and 0 otherwise. One can clearly see that $\mathcal{H}^{\mathcal{F}}_G = \mathcal{O}\widehat{\mathcal{H}}^{\mathcal{F}}_G$, as we account for the sub-graph counts that are disregarded when only considering induced sub-graphs. Since $\mathcal{O}$ is an upper triangular matrix, it is invertible. Thereby, given $\mathcal{H}^{\mathcal{F}}_G$, one can retrieve the induced sub-graph counts by computing $\mathcal{O}^{-1}\mathcal{H}^{\mathcal{F}}_G$. By linearity of expectation, Theorem~\ref{thm:unbiased} implies that the induced sub-graph counts are unbiased as well.

While processing the stream, we store the degree of each vertex, by incrementing the degree for $u_t, v_t$ when $e_t = (u_t,v_t)$ arrives. We use edge-centric algorithms (as described in Section~\ref{sec:est}) to compute estimates for $F_6, F_{13}, \ldots, F_{17}$, and use intuitive combinatorial formulas, listed in Table~\ref{tab:gabeform}, to compute the remaining 11 sub-graphs. We can compute $|E_G|$ and $|V_G|$ by keeping track of how many edges have been received, and the maximum vertex label received, respectively.
\begin{table}[t]
\caption{Graphs and their corresponding sub-graph count formulas.}
\label{tab:gabeform}
\begin{tabularx}{1\textwidth}{cY|cY|cY}
\hline
\textbf{Graph} & \textbf{Formula}                          & \textbf{Graph} & \textbf{Formula}                             & \textbf{Graph} & \textbf{Formula}                             \\ \hline
\parbox[c]{1em}{\centering\includegraphics[width=1em]{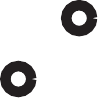}}    & $\binom{|V_G|}{2}$               & \parbox[c]{1em}{\centering\includegraphics[width=1em]{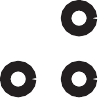}}    & $\binom{|V_G|}{3}$                  & \parbox[c]{1em}{\centering\includegraphics[width=1em]{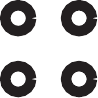}}    & $\binom{|V_G|}{4}$                  \\
\parbox[c]{1em}{\centering\includegraphics[width=1em]{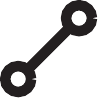}}    & $|E_G|$                          &\parbox[c]{1em}{\centering\includegraphics[width=1em]{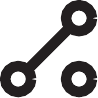}}    & $|E_G|(|V_G|-2)$                    & \parbox[c]{1em}{\centering\includegraphics[width=1em]{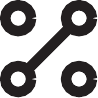}}    & $|E_G|\binom{|V_G|-2}{2}$           \\
\parbox[c]{1em}{\centering\includegraphics[width=1em]{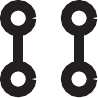}}   & $\binom{|E_G|}{2} - |H^{F_5}_G|$ & \parbox[c]{1em}{\centering\includegraphics[width=1em]{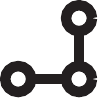}}    & $\sum_{v \in V_G} \binom{d^v_G}{2}$ & \parbox[c]{1em}{\centering\includegraphics[width=1em]{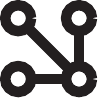}}   & $\sum_{v \in V_G} \binom{d^v_G}{3}$ \\
\parbox[c]{1em}{\centering\includegraphics[width=1em]{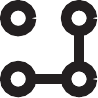}}    & $|H^{F_5}_G|(|V_G|-3)$           & \parbox[c]{1em}{\centering\includegraphics[width=1em]{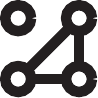}}   & $|H^{F_6}_G|(|V_G|-3)$              & -     & -                                   \\ \hline
\end{tabularx}
\caption{Features extracted for each vertex, $v \in V_G$ for \textsc{maeve}, their formulae, and a figure highlighting the relevant edges. The filled in vertex depicts $v$.}
\label{tab:maeve}
{\def\arraystretch{1.6}
\begin{tabularx}{1\textwidth}{cYYYc}
Degree & Clustering Coefficient & Avg. Degree of $N_G(v)$ & Edges in $I_G(v)$ & Edges leaving $I_G(v)$ \\ \hline
$d^v_G$ & $|T_G(v)|/\binom{d^v_G}{2}$ & $1+|P_G(v)|/d^v_G$ & $d^v_G + |T_G(v)|$ & $|P_G(v)| - 2|T_G(v)|$ \\
\parbox[c]{0.15\textwidth}{\centering\includegraphics[width=0.13\textwidth]{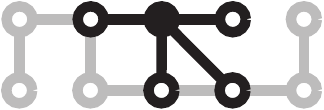}} & \parbox[c]{0.15\textwidth}{\centering\includegraphics[width=0.13\textwidth]{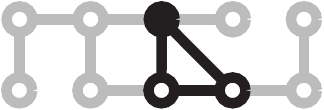}} & \parbox[c]{0.15\textwidth}{\centering\includegraphics[width=0.13\textwidth]{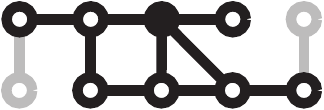}} & \parbox[c]{0.15\textwidth}{\centering\includegraphics[width=0.13\textwidth]{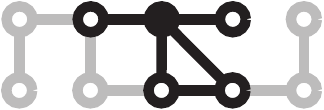}} & \parbox[c]{0.15\textwidth}{\centering\includegraphics[width=0.13\textwidth]{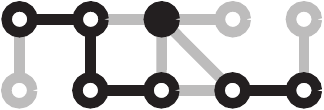}}
\end{tabularx}
}
\end{table}

\vskip.05in 
\noindent
{\bf Time and Space Complexity.} An array of size $|V_G|$ is used to store degrees, which can be accessed in $O(1)$ time, and hence the counts for $F_5$ and $F_{12}$ can be incremented each time an edge arrives in $O(1)$. Let $G'$ denote the graph represented by $\widetilde{E_G}$, stored as an adjacency list. Determining if two vertices are adjacent takes $O(\log b)$ time when using a tree data-structure within the stored adjacency list. At the arrival of $e_t = (u_t, v_t)$, we need to visit only the vertices two hops away from $u_t$ (resp. $v_t$), then perform at most three adjacency checks. Thereby, we perform $2\left(\sum_{w \in N_{G'}(u_t)} d^w_{G'} + \sum_{w \in N_{G'}(v_t)} d^w_{G'}\right) \times 3\log{b} = O(b\log{b})$ operations for one edge, and $O(b\log{b}|E_G|)$ in total. Storing an adjacency list with $b$ edges, and an array for degrees takes $O(b + |V_G|)$ space.

\subsection{Moments of Attributes Estimated on Vertices Efficiently}
\label{sec:maeve}

NetSimile~\cite{berlingerio2013network} propose extracting features for each vertex and aggregating them by taking moments over their distribution. Similarly, we propose extracting a subset of those features, listed in Table~\ref{tab:maeve}, and computing four moments for each feature: mean, standard deviation, skewness, and kurtosis.

\vskip.05in 
\noindent
{\bf Extracting Vertex Features.} For a graph $G$, and a vertex $v \in V_G$, we use $I_G(v)$ to denote the induced sub-graph of $G$ formed by $v$ and its neighbors. Let $T_G(v)$ be the set of triangles that $v$ belongs to, and $P_G(v)$ be the set of three-paths (paths on three vertices) where $v$ is an end-point. We compute the features in Table~\ref{tab:maeve} by using their formulas on estimates of $|T_G(v)|$, $|P_G(v)|$, and $d^v_G$ computed for each $v \in V$ as in Sections~\ref{sec:est} and~\ref{sec:gabe}.

\begin{theorem}
For a vertex $v \in V_G$, all vertex features used in \textsc{maeve} can be expressed in terms of $d^v_G$, $|T_G(v)|$, and $|P_G(v)|$.
\end{theorem}

\begin{proof}
The first two are already expressed in terms of $d^v_G$ and $|T_G(v)|$. 

\noindent\textit{Average Degree of Neighbors:} For each $u \in N_G(v)$, there is exactly one edge connected to $v$, accounting for $d^v_G$ edges. The remaining edges are part of three-paths on which $v$ is an end-point. Therefore, $\sum_{u \in N_G(v)} d^u_G = d^v_G + |P_G(v)|$.

\noindent\textit{Edges in $I_G(v)$:} There are two types of edges in $E_{I_G(v)}$: (1) edges incident on $v$, of which there are $d^v_G$, and (2) edges not incident on $v$. The latter must belong to a pair of vertices which form a triangle with $v$. For each such edge, there is exactly one triangle. Therefore, $\left|E_{I_G(v)}\right| = d^v_G + |T_G(v)|$.

\noindent\textit{Edges leaving $I_G(v)$:} Consider a sub-graph $h \in P_G(v)$. Let $u$ be the other end-point of $h$, and $w$ be the center vertex. When $u \not\in N_G(v)$, it belongs to a three-path that is not in $N_G(v)$, and is thereby an edge leaving the induced sub-graph of $v$. Now, consider $u \in N_G(v)$. Clearly, the edge $(u, w)$ forms a triangle, and is incident in exactly two three-paths: $\{(v, u), (u, w)\}$ and $\{(v, w), (u, w)\}$. Therefore, if we account for the three-paths that lie within $N_G(v)$, we can formulate the number of edges leaving $I_G(v)$ as $|P_G(v)| - 2|T_G(v)|$.
\end{proof}
\vskip-.07in
\noindent {\bf Time and Space Complexity.} As in Section~\ref{sec:gabe}, we assume an adjacency list with an underlying tree structure and refer to the sampled graph as $G'$. At the arrival of an edge $e_t = (u_t, v_t)$, one can traverse the neighborhoods to obtain the triangle and three-path count in $(N_{G'}(u_t) + N_{G'}(v_t)) + N_{G'}(u_t) + N_{G'}(v_t) = O(b)$ time. We store three arrays of size $|V_G|$ to store degrees, triangle counts, and three-path counts. We can compute the moments in at most two passes over these arrays, giving us a total of $O(b|E_G | + |V_G|)$ time. Storing an adjacency list of size $b$ and arrays of size $|V_G|$ gives us $O(b + |V_G|)$ space.

\vskip.05in 
\noindent
{\bf Improving Estimation Quality with Multiple Workers.} Multiple worker machines can be used in parallel to independently estimate triangle counts before aggregating them~\cite{shin2018tri}. Using $W$ worker machines decreases the variances by a factor of $1/W$. Their methodology can be adopted \textit{mutatis mutandis} in our algorithms to improve the estimation quality.


\begin{table}[h]
\caption{Details of classification datasets. The number of graphs, classes, and minimum/maximum number of vertices/edges in a graph have been provided.}
\label{tab:benchmark}
\begin{tabularx}{1\textwidth} {lYYYY}
\hline
\textbf{Dataset}      & $|\mathcal{G}|$ & Classes & $\max|V_G|$ & $\max|E_G|$ \\ \hline
D\&D             & 1,178                          & 2       & 5,748       & 14,267\\
COLLAB           & 5,000                           & 3       & 492         & 40,120\\
REDDIT-BINARY    & 2,000                           & 2       & 3,782       & 4,071\\
REDDIT-MULTI-5K  & 4,999                           & 5       & 3,648       & 4,783\\
REDDIT-MULTI-12K & 11,929                          & 11      & 3,782       & 5,171\\ \hline
\end{tabularx}

\caption{Massive networks with their order, size, and what they represent.}
\label{tab:massive}
\begin{tabularx}{1\textwidth} {lYYY}
\hline
\textbf{Graph} & $|V_G|$    & $|E_G|$     & \textbf{Network Type}    \\\hline
Patent         & 3,774,768  & 16,518,937  & Citation \\
Flickr         & 2,302,925  & 22,838,276  & Friendship       \\
Full USA       & 23,947,347 & 28,854,312  & Road             \\
UK Domain 2002 & 18,483,186 & 261,787,258 & Hyperlink  \\\hline
\end{tabularx}
\caption{Classification accuracy on the datasets described in Table~\ref{tab:benchmark}. Results within 1\% of the best have been boldfaced.}
\label{tab:res}
\begin{tabularx}{1\textwidth} {lYYYYY}
\hline
\textbf{Descriptor} & \textbf{DD} & \textbf{COLLAB} & \textbf{RDT-2} & \textbf{RDT-5} & \textbf{RDT-12} \\ \hline
NetLSD~\cite{tsitsulin2018netlsd} & \textbf{70.36\%} & \textbf{74.27\%} & 82.85\% & \textbf{41.23\%} & 30.9\% \\ \hline
\textsc{gabe} ($b=\nicefrac{1}{4}|E_G|$) & 65.23\% & 63.62\% & 84.65\% & \textbf{41.1\%} & 32.18\% \\
\textsc{gabe} ($b=\nicefrac{1}{2}|E_G|$) & 69.08\% & 65.23\% & \textbf{85.35}\% & 40.63\% & \textbf{32.96\%} \\ \hline
\textsc{maeve} ($b=\nicefrac{1}{4}|E_G|$) & 59.44\% & 68.42\% & 85.04\% & 41.15\% & 32.57\% \\
\textsc{maeve} ($b=\nicefrac{1}{2}|E_G|$) & 61.26\% & 70.95\% & \textbf{86.15}\% & \textbf{41.53\%} & \textbf{33.69\%} \\ \hline
\end{tabularx}
\end{table}
\section{Experimental Evaluation}\label{sec:experiments}
In this section, we perform experiments to show how the approximation quality changes with respect to $b$, explore how the descriptors perform on  classification tasks, and showcase the scalability of the algorithms. 
As in ~\cite{berlingerio2013network}, from extensive experiments, we found that Canberra distance $\left(d(\Vec{x},\Vec{y}) := \sum_{i = 1}^{d} \frac{|\Vec{x}_i - \Vec{y}_i|}{|\Vec{x}_i|+|\Vec{y}_i|}\right)$ performs best when comparing the descriptors. 
We refer to the approximation error as the distance between the true vectors and their approximations.

\vskip.05in
\noindent\textbf{Implementation.} All experiments were performed on a machine with 48 Intel Xeon E5-2680 v3 \@ 2.50GHz Processors, and 125 GB RAM. The algorithms are implemented\footnote{Code: \url{https://github.com/zohair-raza/estimating-graph-descriptors/}} in C++ using MPICH 3.2 on the base code provided by the authors of Tri-Fly~\cite{shin2018tri}. We use 25 processes to simulate 1 Master and 24 workers. Each descriptor is computed exactly once under this setting.
\vskip.05in
\noindent\textbf{Datasets.} We evaluate our models on randomly sampled REDDIT graphs\footnote{\url{https://dynamics.cs.washington.edu/data.html}}, five benchmark classification datasets with large graphs: D\&D, COLLAB, REDDIT-BINARY, REDDIT-MULTI-5K, and REDDIT-MULTI-12K~\cite{yanardag2015deep} (Table~\ref{tab:benchmark}), and massive networks from KONECT~\cite{konect} (Table~\ref{tab:massive}). For each graph, we remove duplicated edges and self-loops, convert to edge-list format with vertex labels in the range $[0, |V_G|-1]$, and randomly shuffle the list.
\subsection{Approximation Quality}
We uniformly sampled 1000 graphs of size 10,000 to 50,000 from REDDIT, representing interactions in various ``sub-reddits''. In Figure~\ref{fig:plots}(a) we show how the average approximation error taken over all the sampled graphs decreases as $b$ (a fraction of the number of edges) increases.



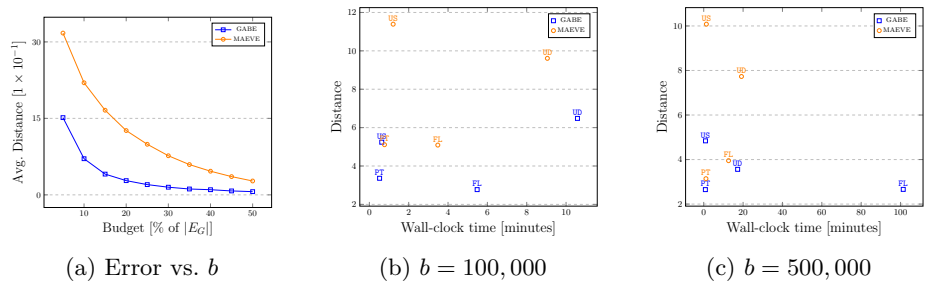
\begin{figure}[h!]
    \centering
    \subcaptionbox{Error vs. $b$}{%
    \resizebox {0.29\textwidth} {!} {%
    \centering
    \begin{tikzpicture}
    \begin{axis}[
        xlabel={Budget [\% of $|E_G|$]},
        ylabel={Avg. Distance [\num{1e-1}]},
        xtick={0,10,20,30,40,50},
        ytick={0,15,30},
        legend pos=north east,
        ymajorgrids=true,
        grid style=dashed,
        label style={font=\Large},
    ]
    \addplot[
        color=blue,
        mark=square,
        ]
        coordinates {
        (5,15.15)(10,7.09)(15,4.07)(20,2.79)(25,2.01)(30,1.5)(35,1.15)(40,1.01)(45,0.77)(50,0.63)
        };
        \addlegendentry{\textsc{gabe}}
    \addplot[
        color=orange,
        mark=o,
        ]
        coordinates {
        (5,31.73)(10,21.98)(15,16.6)(20,12.61)(25,9.92)(30,7.67)(35,5.94)(40,4.64)(45,3.58)(50,2.72)
        };
        \addlegendentry{\textsc{maeve}}
    \end{axis}
    \end{tikzpicture}
    }}%
    \hfill
    \subcaptionbox{$b=100,000$}{%
    \resizebox {0.30\textwidth} {!} {%
    \centering
    \pgfkeys{
    /prepare label/.style={
        /print label/\detokenize{#1}/.code={\ttfamily\detokenize{#1}}
    },
    /prepare label/.list={PT,FL,US,UD}
    }
    \begin{tikzpicture}
    \begin{axis}[
        xlabel={Wall-clock time [minutes]},
        ylabel={Distance},
        legend pos=north east,
        ymajorgrids=true,
        label style={font=\Large},
        grid style=dashed
    ]
        \addplot[
            mark=square,
            color=blue,
            only marks,
            point meta=explicit symbolic,
            nodes near coords={
                \pgfkeys{/print label/\pgfplotspointmeta/.try}
            }
        ]
            table[header=false,meta index=0,x index=1,y index=2, row sep=crcr]{%
                PT	0.5203333333	3.36\\
                    FL	5.48	2.77\\
                    US	0.63	5.24\\
                    UD	10.58	6.48\\
            };
        \addlegendentry{\textsc{gabe}}
        \addplot[
            mark=o,
            color=orange,
            only marks,
            point meta=explicit symbolic,
            nodes near coords={
                \pgfkeys{/print label/\pgfplotspointmeta/.try}
            }
        ]
            table[header=false,meta index=0,x index=1,y index=2, row sep=crcr]{%
                PT	0.77	5.11\\
                FL	3.48	5.09\\
                US	1.21	11.39\\
                UD	9.05	9.61\\
            };
        \addlegendentry{\textsc{maeve}}
    \end{axis}
    \end{tikzpicture}
    }}%
    \hfill
    \subcaptionbox{$b=500,000$}{%
    \resizebox {0.30\textwidth} {!} {%
    \centering
    \pgfkeys{
    /prepare label/.style={
        /print label/\detokenize{#1}/.code={\ttfamily\detokenize{#1}}
    },
    /prepare label/.list={PT,FL,US,UD}
    }
    \begin{tikzpicture}
    \begin{axis}[
        xlabel={Wall-clock time [minutes]},
        ylabel={Distance},
        legend pos=north east,
        ymajorgrids=true,
        label style={font=\Large},
        grid style=dashed
    ]
        \addplot[
            mark=square,
            color=blue,
            only marks,
            point meta=explicit symbolic,
            nodes near coords={
                \pgfkeys{/print label/\pgfplotspointmeta/.try}
            }
        ]
            table[header=false,meta index=0,x index=1,y index=2, row sep=crcr]{%
                PT	0.84	2.65\\
                FL	101.37	2.66\\
                US	0.90	4.84\\
                UD	17.23	3.56\\
            };
        \addlegendentry{\textsc{gabe}}
        \addplot[
            mark=o,
            color=orange,
            only marks,
            point meta=explicit symbolic,
            nodes near coords={
                \pgfkeys{/print label/\pgfplotspointmeta/.try}
            }
        ]
            table[header=false,meta index=0,x index=1,y index=2, row sep=crcr]{%
                PT	1.13	3.14\\
                FL	12.62	3.95\\
                US	1.34	10.08\\
                UD	19.18	7.73\\
            };
        \addlegendentry{\textsc{maeve}}
    \end{axis}
    \end{tikzpicture}
    }}%
    \caption{\small Approximation error and runtime of \textsc{gabe} and \textsc{maeve} (best viewed in color).}
    \label{fig:plots}
\end{figure}

\subsection{Graph Classification}
We computed descriptors for graphs in Table~\ref{tab:benchmark} from samples of 25\% and 50\% of all the edges and examined their classification accuracy. We used the state-of-the-art descriptor, NetLSD~\cite{tsitsulin2018netlsd}, as a benchmark, despite the fact that our models have {no direct competitors}. As in~\cite{tsitsulin2018netlsd}, we used a simple 1-Nearest Neighbour classifier. We performed 10-fold cross-validation for 10 different random splits of the dataset (i.e. 100 different folds are tested on), and report the average accuracy in Table~\ref{tab:res}. Note that despite using only a fraction of edges, \textsc{gabe} and \textsc{maeve} give results competitive to the state of the art.
\subsection{Scaling to Large Real-world Networks} 
We run our algorithms on massive networks (Table~\ref{tab:massive}) and estimated descriptors by setting $b$ to $100,000$ and $500,000$. In Figures~\ref{fig:plots}(b) and (c), we show the scatter plots for wall-clock time taken vs. the distance between the real vectors and their approximations (values nearer to the origin are better). We are able to process a graph with $\approx 260$ million edges under 20 minutes, with relatively low error. Note that when $b = 500,000$, \textsc{gabe} takes 102 minutes to compute the descriptor for Flickr, implying that we must take the density of the graph into account for efficient computation when setting the value of $b$.
\section{Conclusion}
\label{sec:conclusion}
In this work, we present single-pass streaming algorithms to construct graph descriptors using a fixed amount of memory. We show that these descriptors provide better approximations with increasing $b$, are comparable with the state-of-the-art known descriptors in terms of classification accuracy, and scale well to networks with millions of vertices and edges. 
%
%
\bibliographystyle{splncs04}
\bibliography{main.bib}
\end{document}